\documentclass[final, 10pt, twocolumn, conference]{IEEEtran}
\usepackage{setspace}
\usepackage{amsmath}
\usepackage{amsthm}
\usepackage{amssymb,enumitem}
\usepackage{algorithm}
\usepackage[noend]{algpseudocode}
\usepackage{cite}
\usepackage{graphicx}
\usepackage{epstopdf}
\usepackage{url}
\usepackage{cite}
\usepackage{texdef2015}
\usepackage{color}
\usepackage{caption}
\usepackage{subcaption}
\usepackage{float}
\usepackage{balance}
\allowdisplaybreaks

\usepackage{tikz}
\usetikzlibrary{automata,arrows,positioning,calc,fit,shapes.multipart,chains,shapes}
\usepackage{microtype}

\newtheorem{theorem}{Theorem}
\newtheorem{corollary}{Corollary}

\newtheorem{lemma}{Lemma}

\newtheorem{algo}{Algorithm}

\newcommand{\age}{\Delta}

\newcommand{\negfigspace}{\vspace{-2mm}}

\newcommand{\negspacesmall}{\vspace{-1em}}

\newcommand{\paoi}{\Delta^P} 

\makeatletter
\def\BState{\State\hskip-\ALG@thistlm}
\makeatother

\newcommand{\QueueR}[2]{\draw[thick] (#1) ++(-0.7,0.3) --  ++(0.7,0) -- ++(0,-0.6) --  ++(-0.7,0);
\foreach \i in {1,...,3}
	\draw[thick] (#1) ++(-\i*0.2,0.3) -- ++(0,-0.6);
\coordinate (#2) at ($(#1) + (-0.7,0)$);}
\newcommand{\blen}{N} 

\newcommand{\enc}{\mathcal{E}}
\newcommand{\dec}{\mathcal{D}}
\newcommand{\Xset}{\mathcal{X}}
\newcommand{\ns}{\epsilon} 

\begin{document}

\title{Timely Lossless Source Coding for Randomly Arriving Symbols}
\author{Jing Zhong, Roy D.~Yates and Emina Soljanin \\
\small Department of ECE, Rutgers University, \{jing.zhong, ryates, emina.soljanin\}@rutgers.edu}

\maketitle

\begin{abstract}
\boldmath 
We consider a real-time streaming source coding system in which an encoder observes a sequence of randomly arriving symbols from an i.i.d. source, and feeds binary codewords to a FIFO buffer that outputs one bit per time unit to a decoder.
Each source symbol represents a status update by the source, and the timeliness of the system is quantified by the age of information (AoI), defined as the time difference between the present time and the generation time of the most up-to-date symbol at the output of the decoder.
When the FIFO buffer is allowed to be empty, we propose an optimal prefix-free lossless coding scheme that minimizes the average peak age based on the analysis of discrete-time Geo/G/1 queue.
For more practical scenarios in which a special codeword is reserved for indicating an empty buffer, we propose an encoding scheme that assigns a codeword to the empty buffer state based on an estimate of the buffer idle time.

\end{abstract}

\section{Introduction}

Many ubiquitous computing applications share a common need: the information update from the source has to be available at the interested receivers as quickly as possible.
A recently developed timeliness metric, the \emph{age of information} (AoI), quantifies the information freshness of status updating systems \cite{Kaul2012infocom, Costa2014,Huang2015,Sun2016,Najm2017,Bedewy2016,Kadota2016,Yates2017}.
More specifically, age measures the time difference between now and when the most recent update was generated.
If the receiver receives an update at some time $t$, and an update was generated at time $u(t)$, then the instantaneous age at the receiver is $t-u(t)$.


Real-time communication systems, such as live video streaming and information update in vehicular networks, often require efficient compression that enables the receiver to reconstruct the source message in a timely manner under limited network resources.
The analysis of these systems can be simplified to a real-time compression problem over a constrained data network.
In this work, we restrict our attention to the following baseline problem: if every update by the source is transmitted to the receiver through a binary channel with a fixed rate, what is optimal compression scheme that keeps the information about the source at the receiver as timely as possible?
This problem is different from the traditional source coding  that focuses on minimizing the average codeword length in order to approach the Shannon entropy of the source.
 
The delay of streaming source coding has been studied in different contexts. 
The end-to-end delay of streaming source coding was first studied in \cite{Humblet1978}. Here, source symbols arrive as a Poisson process, and the encoder maps them into binary codewords and puts them in a finite size buffer that outputs one bit per time unit.
A variant of the Huffman code was proposed to minimize the probability of buffer overflow.
A similar problem was studied in \cite{Chang2006}, in which source symbols arrive at the encoder sequentially one per time unit, and the receiver is required to reconstruct the source with a fixed end-to-end delay constraint.
It is necessary to distinguish our timeliness requirement from measuring the end-to-end delay in \cite{Humblet1978} and \cite{Chang2006}, since the age is a process that captures how old the information about the source is at the receiver.

Our prior work \cite{Zhong2016} applied age analysis to a streaming source coding system with a deterministic source symbol inter-arrival times. 
We assumed that a prefix-free fixed-to-variable encoder maps every block of $\blen$ symbols to a binary codeword that is sent through a bit pipe that outputs $R$ bits per time unit.
We observed that the encoder must choose an appropriate blocklength $\blen$ to balance data compression delays against network congestion deriving from insufficient compression.
Given a blocklength $\blen$, we proposed a coding scheme to optimize average age.
In \cite{Zhong2017},  age analysis was extended to a backlog-adaptive source coding model that makes the busy/idle state at the channel interface available at the source encoder. 
This enables the encoder to adjust the blocklength $\blen$ based on the state of the channel.
In \cite{Mayekar2018}, each source symbol represents a timely update message sent by the source, but the symbols that arrive at the encoder while the channel is busy are skipped.
An optimal Shannon code was proposed to minimize the average age of the freshest source symbol at the receiver.

\begin{figure*}[t]
\centering\small
\begin{tikzpicture}[node distance=1cm]
\node [draw,circle, rounded corners,align=center,minimum height=15mm,thick] (newsource) {Source};
\node[draw,rectangle,rounded corners,align=center,minimum height=12mm,thick] (encoder)[right = 2 of newsource] {Encoder};
\draw[->,thick] (newsource.east) --  node [above] {$X_1X_2X_3\cdots$} (encoder.west);
\node[draw,circle,thick] (FIFO) [right = 2.5 of encoder]{\small$ R=1$}; 
\node (bitpipe) at (FIFO.west) {};
\QueueR{bitpipe}{Bitpipe}
\draw[->,thick] (encoder.east) --  node [above] {$10110\cdots$} (Bitpipe);
\node[draw,rectangle,rounded corners,align=center,minimum height=12mm,thick] (decoder)[right = 2 of FIFO] {Decoder};
\draw[->,thick] (FIFO.east) --  node [above] {$10110\cdots$} (decoder.west);
\node[draw,circle,rounded corners,align=center,minimum height=15mm,thick] (monitor)[right = 2 of decoder] {Receiver};
\draw[->,thick] (decoder.east) --  node [above] {$X_1X_2X_3\cdots$} (monitor.west);
\end{tikzpicture}
\caption{System diagram for streaming source coding.}
\label{fig:source_sys}
\negfigspace
\end{figure*}
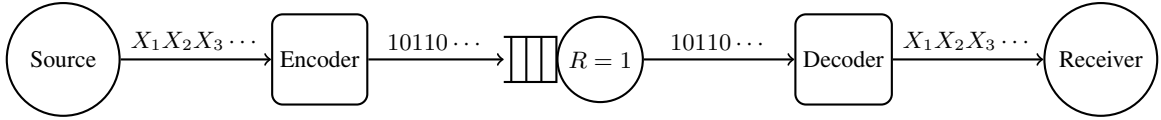

\begin{figure}[t]
\centering
\includegraphics[width=0.45\textwidth]{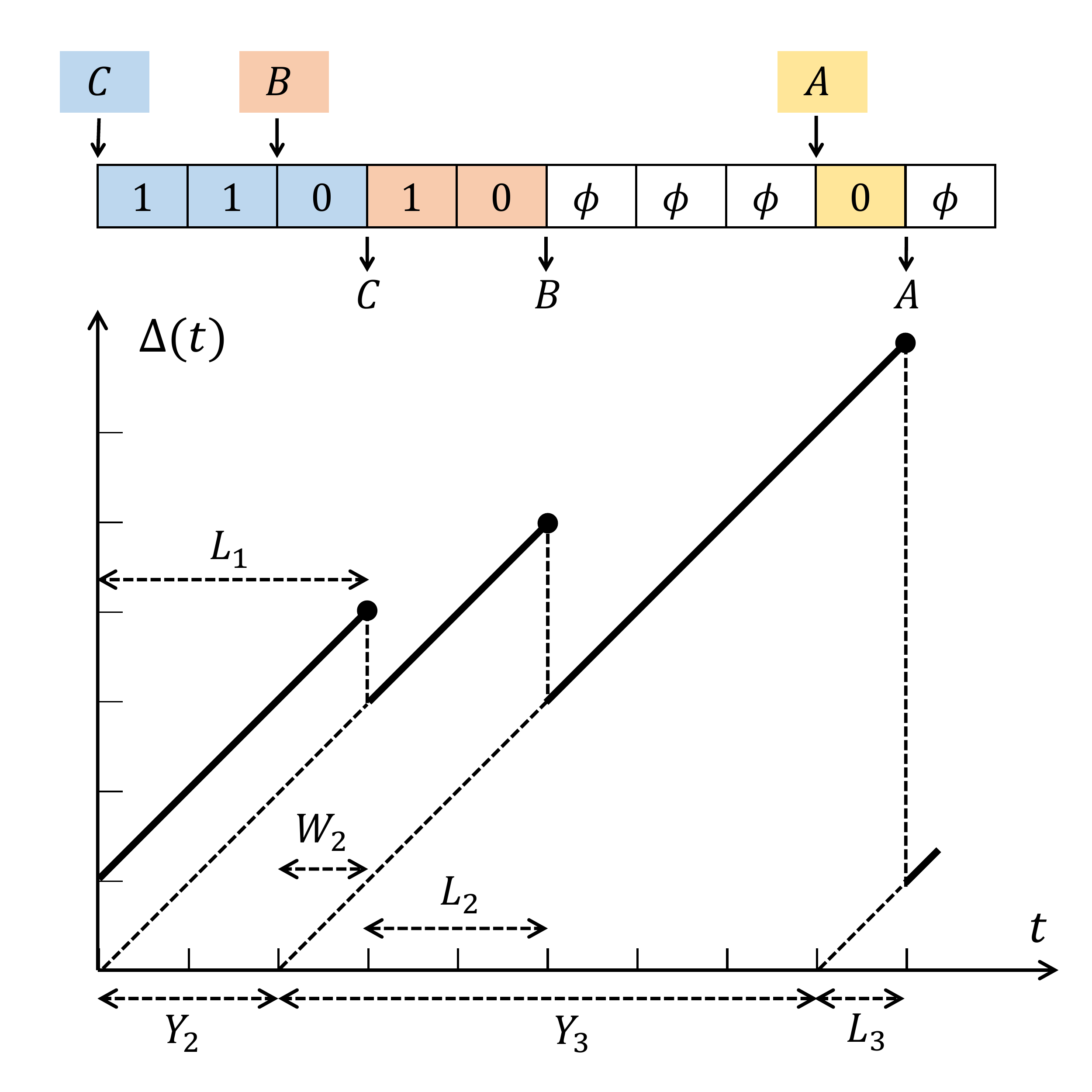}
\caption{An example of the output process of the FIFO buffer and the corresponding age process for the prefix-free lossless coding on random updates. A special symbol $\phi$ is transmitted when the FIFO buffer is empty.}
\label{fig:ideal}
\negspacesmall
\end{figure}

In this paper, we consider the discrete-time streaming source coding system with random arrivals shown in Fig.~\ref{fig:source_sys}.
This system differs from other systems with deterministic symbol arrivals in \cite{Zhong2016,Zhong2017,Mayekar2018}. 
Here we assume a source symbol arrives as a Bernoulli process with probability $q$ at each time unit.
Unlike other status updating systems in which only the freshness of the most recent update matters, here we require the receiver to reconstruct the entire source message stream in a lossless manner.
Our objective is to design a lossless coding scheme that minimizes the average peak age for randomly arriving source symbols.

We start in Sec.~\ref{sec:ideal} with an idealized system model that provides an empty buffer signal to tell the decoder when the channel buffer is empty.
A prefix-free coding scheme is proposed to minimize the average peak age.
In practical settings, however, the source or channel has to encode and the empty buffer state for the decoder. 
We then 
investigate possible encoding schemes for the empty buffer state in Sec.~\ref{sec:protocol}.
We propose a predictive scheme that assigns a codeword to the empty buffer message based on an estimate of the fraction of time the buffer is idle.
A numerical comparison between different empty buffer encoding schemes is provided in Sec.~\ref{sec:evaluate}.

\section{Age Analysis with Empty Buffer Signaling} \label{sec:ideal}

Consider the streaming source coding system shown in Fig.\ref{fig:source_sys}.
In each time slot (starting from $t=1$), the source is either idle or it generates 
a discrete i.i.d.\ symbol $X_k$ from a finite alphabet $\mathcal{X}$. The source is not idle with probability $q$. 
Each symbol $X_k$ has PMF $P_X(x)$ and is time-stamped when it is observed by the encoder.  
Let $N(t)$ denote the number of symbols observed by the encoder by time $t$.

The lossless source encoder $\enc$ maps every symbol into a prefix-free binary sequence, i.e. $\enc: \Xset \to \{0,1\}^*$, and feeds the encoded sequence $\enc(X_i)$ into a first-in-first-out (FIFO) buffer that outputs $R=1$ bit per time unit. 
The capacity of the FIFO buffer is assumed to be infinite.
A symbol $X_k$ is declared at the output of decoder $\dec$ only after the entire bit sequence $\enc(X_k)$ is delivered to the input of $\dec$.
At every time $t$, the decoder reconstructs the source sequence up to $X^{N(u(t))}$, where $u(t)<t$ is the time stamp of the most recent decoded source symbol. 
We note that $u(t)$ is advanced to a new time index only if a new symbol is decoded.
The age of the source sequence $X^{N(u(t))}$ at the receiver at time $t$ is then given by $	\age(t) = t-u(t)$.

In this section, we assume that a special signal $\phi$ is sent through the channel to indicate to the decoder that the buffer was empty and no symbol arrived at the encoder one time unit prior to receiving $\phi$. 
Denoting $L_k=l(X_k)$ as the encoded bit sequence length of a symbol $X_k$, then the sequence $L_k$ is also i.i.d.\ with PMF $P_L(l)$.  

\begin{table}[t]
\centering
\caption{Example of prefix-free codebook with $|\Xset|=4$.}
\label{table:ideal}
\begin{tabular}{|l|c|c|c|c|}
\hline
$X$     & A & B & C & D  \\ \hline
$\enc(X)$ & 0 & 10 & 110 & 111 \\ \hline
\end{tabular}
\end{table}

Fig.~\ref{fig:ideal} depicts an example of the FIFO buffer output process and the age process. 
Source symbols $X\in\{A,B,C,D\}$ arrive at the input of the encoder sequentially, and each symbol is encoded using the prefix-free codebook specified in Table~\ref{table:ideal}.
The first symbol $X_1=C$ arrives at time $t=0$, and the corresponding bit sequence $110$ is fed into the FIFO buffer and output to the decoder after $L_1=3$ time units. 
Thus, the age $\age(t)$ increases linearly from an initial value $\age_0=1$ and drops to $\age(3)=3$ time units at time $t=3$.  
The second symbol $X_2=B$, which arrives at time $t=2$, is deferred by one time unit since the buffer is serving the codeword for the previous symbol $C$, and delivered to the decoder at time $t=5$.
The age is then reset to the waiting time plus the codeword transmission time for symbol $B$.
Afterwards, the buffer stays empty since there is no new arriving symbol.
The instantaneous age $\age(t)$ increases linearly, and is reset to $L_3=1$ time unit only after the decoder receives the codeword for the third symbol $X_3 = A$.

From a queueing perspective, we can view each source symbol $X_k$ as an arriving job to the system.
The service time $S_k$ of the job $X_k$ is then the time it takes to be transmitted to the decoder, which is exactly the length of the encoded sequence $L_k$.
Thus the expected service time is $\E{S} = \E{L} = 1/\mu$.
The job interarrival time $Y_k$ is geometrically distributed with PMF $P_Y(y) = (1-q)^{y-1} q$ for all $k$, and thus the arrival rate is $\lambda = 1/\E{Y} = q$. 
Since the system behaves as a discrete-time Geo/G/1 queue, we have the following claim.

\begin{lemma}
	The queue is stable if and only if $\E{L}<1/q$.
\end{lemma}

Note that the average codeword length $\E{L}$ is lower bounded by the entropy of the source $H(X)$. Hence, it is necessary to have source entropy $H(X)<1/q$ for a stable queue.

We denote $\age_k$ as the $k$-th peak value of the age process $\age(t)$.
The average peak age (peak AoI) at the receiver is then defined as \cite{Costa2014}
\begin{align}
	\paoi = \lim_{K\to\infty} \frac{1}{K} \sum_{k=1}^{K} \age_k.
\end{align}

\begin{theorem}
	For a stable streaming source coding system with code length distribution $P_L(l)$, the PAoI is given by
	\begin{align}
		\paoi & = \frac{\E{L^2}-\E{L}}{2(1/q-\E{L})} + \E{L} + \frac{1}{q}.
	\end{align} \label{thm:paoi}
\end{theorem} 
\begin{proof}
	Evaluating Fig. \ref{fig:ideal} yields
	\begin{align}
		\paoi 
		& = \lim_{K\to\infty} \frac{1}{K} \sum_{k=1}^{K} (W_k + S_k + Y_k) \nn
		& = \E{W}+\E{S}+ \E{Y}, \label{eqn:EWESEY}
	\end{align}
	where $W_k$ and $S_k$ are the waiting time and service time for source symbol $X_k$ as shown in Fig.~\ref{fig:ideal}.
	We note that the expected waiting time for discrete-time Geo/G/1 queue is given by \cite{Meisling1958}
	\begin{align}
		\E{W} & = \frac{\lambda \E{S(S-1)}}{2(1-\rho)},
	\end{align}
	where $\rho = \lambda/\mu = \E{S}/\E{Y}$ is the system offered load.
	Since $S_k = L_k$ for any $k$, the expected waiting time and service time is rewritten as
	\begin{align}
		\E{W} = \frac{\E{L^2}-\E{L}}{2(1/q-\E{L})}. \label{eqn:EWES}.
	\end{align}
	Theorem \ref{thm:paoi} follows by substituting \eqref{eqn:EWES} into \eqref{eqn:EWESEY}.
\end{proof}

\begin{figure}[t]
\centering
\includegraphics[width=0.45\textwidth]{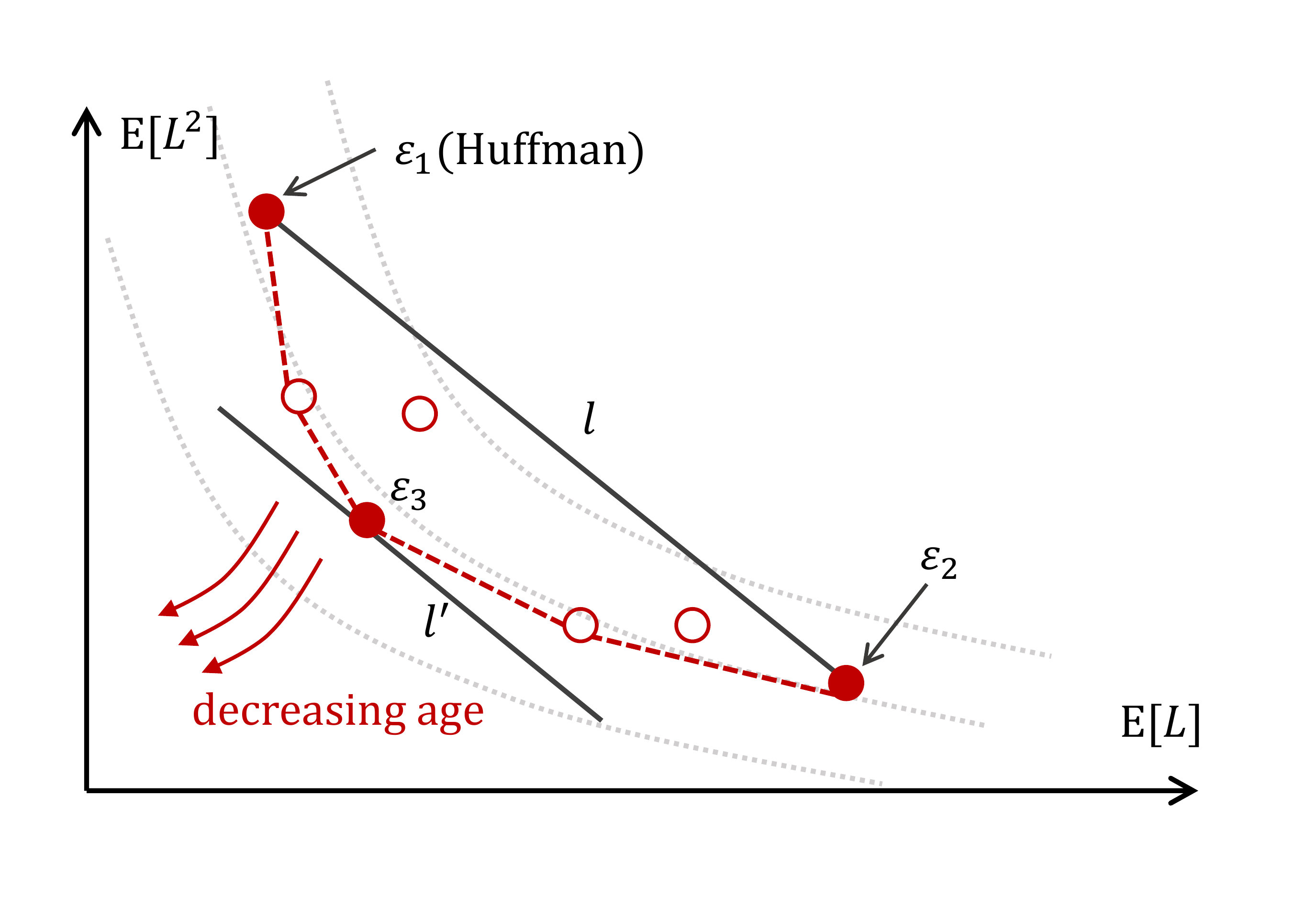}
\caption{The illustration of convex hull algorithm and the representation of codebooks in the coordinate.}
\label{fig:hull}
\negspacesmall
\end{figure}

We observe the PAoI in Thm.~\ref{thm:paoi} is a function of both the average code length $\E{L}$ and the second moment $\E{L^2}$.
This is similar to the bounds on the average age when encoding deterministic arriving source symbols using lossless block-to-variable codes in \cite{Zhong2016}.

\begin{corollary} \label{thm:opt_q}
	For a given source $X$ and encoder $\enc$ with moments of the codeword length $\E{L}$ and $\E{L^2}$, the optimal arrival rate $q^*$ that minimizes the PAoI satisfies
	\begin{align}
		\frac{1}{q^*} = \sqrt{\frac{\E{L^2}-\E{L}}{2}} + \E{L}.
	\end{align}
	The corresponding PAoI is given by
	\begin{align}
		\paoi(q^*) = \sqrt{2(\E{L^2}-\E{L})} + 2\E{L}.
	\end{align}
\end{corollary}
Corollary \ref{thm:opt_q} follows by letting $z=1/q$ and setting $\partial \age^P / \partial z = 0$ and $\E{L}<z$ in Theorem \ref{thm:paoi}.

Next, we use the technique from \cite{Larmore1989,Baer2006} to obtain the optimal coding scheme that minimizes the PAoI in Theorem \ref{thm:paoi} given an arrival rate $q$.
We refer to this coding scheme as the \emph{age optimal code}.
It was shown in \cite{Larmore1989} that all the possible prefix-free lossless codebooks form a convex hull in a two-dimensional space with bases $\{\E{L}, \E{L^2}\}$.
Fig. \ref{fig:hull} depicts an example of the space with the convex hull formed by all possible codebooks.
The goal is to search all the codebooks at the boundary of the convex hull. 
Since $\paoi$ in Thm. \ref{thm:paoi} is convex in $\E{L}$ and $\E{L^2}$, we perform the search by first defining a linear function
\begin{equation}
f(L)=\alpha \Eop[L] + \beta \Eop[L^2], \label{eqn:linearfunc}
\end{equation}
and vary the parameters $\alpha, \beta\in[0,1]$.  
The problem is then reduced to an inner sub-problem of finding the codebook that minimizes the linear penalty function in \eqref{eqn:linearfunc}.
In \cite{Baer2006}, this sub-problem is shown to be reduced to a coin collector's problem, which can be solved recursively by a Package-Merge algorithm \cite{Larmore1990} in linear space and $O(|\mathcal{X}|^2)$ time. 

Given that the inner sub-problem can be solved efficiently, the outer problem is solvable by an iterative algorithm that starts from two extreme cases: $(\alpha,\beta)=(1,0)$ and $(\alpha,\beta)=(0,1)$. 
We remark that $(\alpha,\beta)=(1,0)$ corresponds to a penalty function that returns a prefix code $\enc_1$ that minimizes the average code length, which is a Huffman code. 
Given any two codebooks $\enc_1$ and $\enc_2$, the values of $\alpha$ and $\beta$ are updated as follows
\begin{align}
\alpha & \leftarrow \Eop[L^2](\enc_1) - \Eop[L^2](\enc_2) \label{eqn:alpha} \\
\beta & \leftarrow \Eop[L](\enc_2) - \Eop[L](\enc_1) . \label{eqn:beta}
\end{align}
Next, we find the optimal code $\enc_3$ that minimizes \eqref{eqn:linearfunc} for the new values of $\alpha$ and $\beta$.
Graphically, this step is equivalent to drawing a line segment $l$ that connects the points corresponding to $\enc_1$ and $\enc_2$, and then searching for the lowest line $l'$ parallel to $l$ that touches the boundary of the convex hull consisting of all possible codebooks. 
If $l'$ lies below $l$, then a new codebook $\enc_3$ is contained in the line $l'$. 
This algorithm repeats iteratively by renewing the value of $\alpha$ and $\beta$ by \eqref{eqn:alpha} and \eqref{eqn:beta} at each step for the line segments connecting any two consecutive codebooks, until we find all the feasible codes at the boundary.
The details of this algorithm can be found in \cite{Larmore1989}.

\section{Encoding the Empty Buffer State} \label{sec:protocol}

\begin{figure}[t]
\centering
\includegraphics[width=0.4\textwidth]{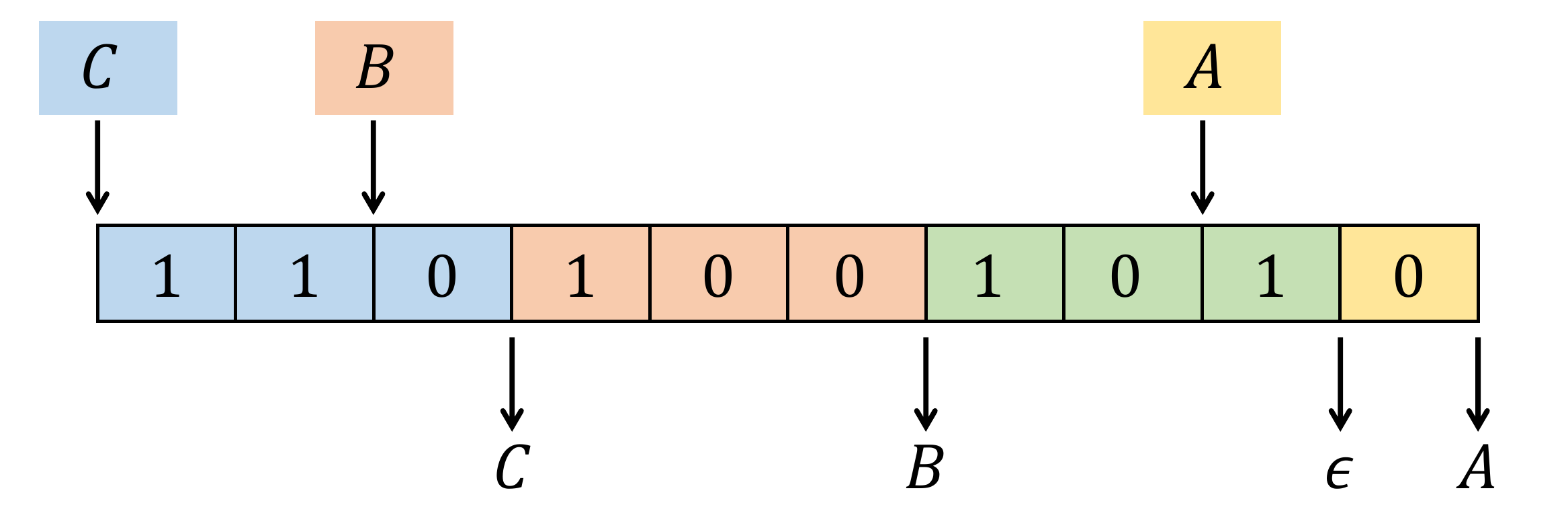}
\caption{An example of the output process of the FIFO buffer in which the special codeword reserved for empty buffer signal $\ns$ is $101$. }
\label{fig:protcol}
\negspacesmall
\end{figure}

\begin{table}[t]
\centering
\caption{Example of codebook with empty buffer symbol $\ns$.}
\label{table:protocol}
\begin{tabular}{|l|c|c|c|c|c|}
\hline
$X$     & A & B & $\ns$ & C & D  \\ \hline
$\enc(X)$ & 0 & 100 & 101& 110 & 111 \\ \hline
\end{tabular}
\end{table}

In this section, we extend the age analysis to a more realistic model where the null symbol $\phi$ is not allowed and the channel can only output either bit 0 or 1. 
In this case, the source has to send a special codeword, which differs from the codewords for the source symbols $X_k$, to inform the decoder when the buffer is empty. 
We refer to this special message as the null symbol $\ns$.
The simplest way is to transmit a single bit ``0'' if the buffer is empty, and otherwise transmit a ``1'' followed by an encoded message.
In this scheme, the ``0'' bit occupies only one time slot and thus doesn't affect the next incoming source symbol.
However, the length of every encoded sequence is increased by 1, and we denote the new length of the message codeword as $L_X = L+1$.
Substituting $L_X$ back into Theorem \ref{thm:paoi} yields the following new PAoI 
\begin{align}
	\paoi_{\textrm{Naive}} & = \frac{\E{L^2}+\E{L}}{2(1/q-\E{L}-1)} + \E{L} +1 + \frac{1}{q}.
\end{align}
We refer to this scheme as the \emph{naive} scheme.
For the naive scheme, the system is stable if and only if $\E{L}+1<1/q$.
That is, for sources with entropy $H(X)\geq1/q-1$, there is no feasible source code for a stable system.

When the system is mostly idle, i.e. $q\ll E[L]$, the buffer has to send the empty state protocol codeword ``0'' frequently. 
Hence, assigning the shortest codeword to $\ns$ is the optimal strategy.
However, when the system is busy all the time, the advantage of a short codeword for empty state will be forfeited since the buffer is overloaded by serving longer codewords for source messages.
Reserving a codeword for an empty buffer state is equivalent to adding a null symbol $\ns$ to the source alphabet in the codebook. 
We then denote the probability of the null symbol used in the codebook as $p_\ns$.
Our objective is to find the optimal $p_\ns$ such that the PAoI is minimized.

Fig.~\ref{fig:protcol} depicts an example of the FIFO buffer output process with the same arrival process as Fig. \ref{fig:ideal} and an alternative empty buffer encoding scheme as shown in Table \ref{table:protocol}.
When the FIFO buffer becomes empty, the codeword $101$ is transmitted to the decoder.
In this case, the codeword corresponding to symbol $A$ is deferred by $1$ time slot since the buffer is busy sending the last bit $1$ of the codeword corresponding to the null symbol $\ns$.

Let $I_\enc$ denotes the fraction of time that the buffer is empty when the codebook $\enc$ is used to compress the source.
One would expect that the choice of $p_\ns$ for encoding the null symbol should be matched to  $I_\enc$. 
However, this is in general not true since $I_\enc$ is the consequence of the encoding scheme $\enc$. 
Consider an example in which the age optimal coding scheme is first applied to generate the codebook for the source symbols, and later the naive scheme is used to include the encoding of the empty buffer state. 
The complete codebook with $\ns$ is denoted by $\enc_1$.
For an infinite length source sequence, the fraction of time that the buffer stays empty is the complement of the offered load, which is denoted by $I_{\enc_1} = 1-q(\E{L}+1)$.
Suppose now the encoder chooses an alternative coding scheme $\enc_2$ that assigns the null symbol $\ns$ with probability $p_\ns = I_{\enc_1}$.
That is, the probability of every source symbol $P(X)$ is scaled by $1-p_\ns$, and the length of the codeword is very likely to be different. 
The changes to both the codeword length and the length of the empty buffer codeword will potentially lead to a new fraction of buffer idle time $I_{\enc_2}\neq I_{\enc_1}$. 

Although it's difficult to obtain the optimal $p_\ns$ that minimizes the PAoI, it would be reasonable to choose $p_\ns$ based on an estimate of the fraction of buffer empty time $I_\enc$.
We propose a simple \emph{predictive} scheme which exploits the buffer offered load $\rho$ when empty buffer signaling is allowed.
The fraction of time for the empty buffer $I_\enc$, which is $1-\rho=1-q\E{L}$ in this case, is then used as $p_\ns$ for encoding.
The detailed procedure for the predictive scheme is shown as follows.
\begin{algo}[Predictive Encoding]
	\qquad 
	\begin{enumerate}
		\item Obtain the PAoI-optimal code $\enc$ assuming empty buffer signaling is allowed, i.e. $\enc\{\ns\} = \phi$. Denote the average codeword length as $\E{L(\enc)}$.
		\item Set the null symbol probability $p_\ns=1-q\E{L(\enc)}$ and set an alternative source PMF $P_{X'}(x) = (1-p_\ns)P_X(x)$ for all $x\in\mathcal{X}$ and $P_{X'}(\ns) = p_\ns$.
		\item Generate the PAoI-optimal codebook for the source $X'$ with PMF $P_{X'}(x)$.
	\end{enumerate}
\end{algo}

When $q\E{L} \ll 1$, $p_\ns$ is large and close to 1, this predictive scheme is identical to the naive scheme since the encoder assigns the most probable codeword to $p_\ns$, which is a single bit $0$ or $1$.
We note that the age analysis for the predictive scheme is relatively complicated since the waiting time of a symbol includes the time waiting for the service of a possible previous null symbol $\ns$ as shown in the example in Fig. \ref{fig:protcol}.
Since every null symbol $\ns$ is inserted in the channel once the buffer becomes empty, the arrival of $\ns$ depends on the buffer state and thus the effective arrival process is not i.i.d..

\section{Evaluations} \label{sec:evaluate}

\begin{figure}[t]
\centering
\includegraphics[width=0.45\textwidth]{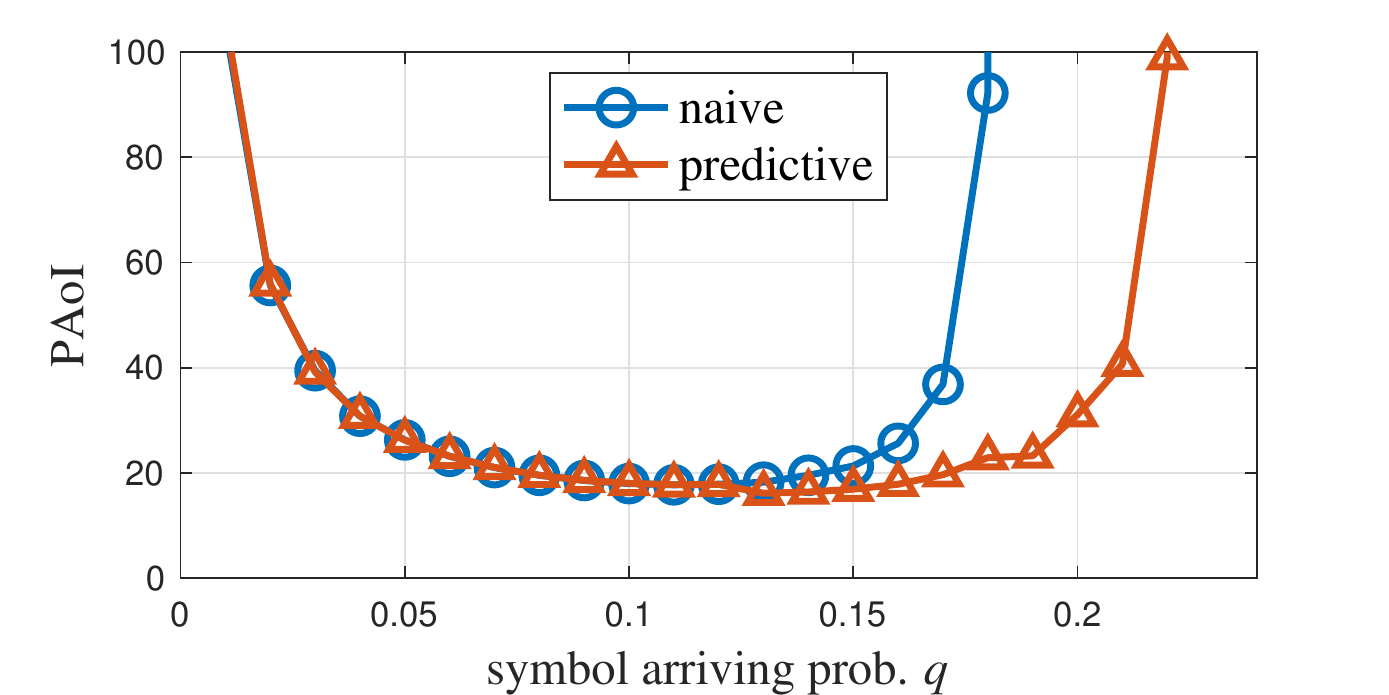}
\caption{PAoI versus arriving probability $q$ for uniform $X$. }
\label{fig:unif}
\negspacesmall
\end{figure}

\begin{figure}[t]
\centering
\includegraphics[width=0.45\textwidth]{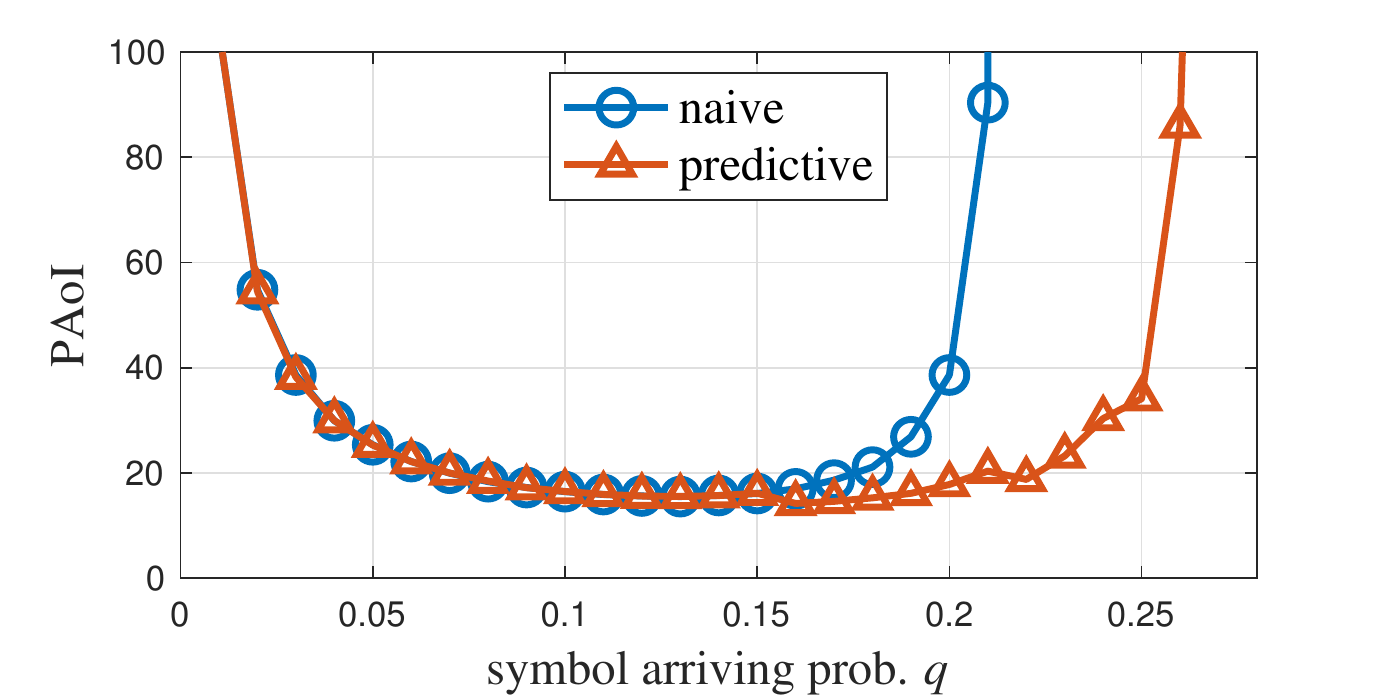}
\caption{PAoI versus arriving probability $q$ for Zipf $X$. }
\label{fig:zipf}
\negspacesmall
\end{figure}

Fig.~\ref{fig:unif} and \ref{fig:zipf} depict PAoI for the two different empty buffer encoding schemes by varying the symbol arriving rate $q$ between $0$ and $1/H(X)$. 
In Fig. \ref{fig:unif}, the source $X$ has $20$ symbols and all the symbols are uniformly distributed with $P_X(x) = 1/20$.  
For any given symbol arrival rate $q$, the naive scheme first finds the age-optimal code when the empty buffer signaling is allowed, and then pads a bit $1$ before sending every message codeword.
When the source arrival rate $q$ is small, both scheme yield large PAoI, and the predictive scheme is identical to the naive scheme as expected since the system is mostly idle.
In this case, the optimal encoding scheme is to assign the shortest codeword to the null symbol $\ns$.
As $q$ increases, the PAoI first decreases and then begins to rise since the system becomes unstable when the average length of message codeword $\E{L}>1/q$.
The curve corresponding to the predictive scheme blows up later than that of the naive scheme when the system load becomes large.  
This is mainly because the predictive scheme assigns a longer codeword to the null symbol $\ns$ and this shortens the average length of the message codeword $\E{L_X}$.

Fig.~\ref{fig:zipf} shows a similar experiment for source $X$ following the Zipf distribution with PMF 
\begin{align*}
	P_X(x) = \frac{1/x^s}{\sum_1^{n}1/x^s},	
\end{align*}
where we set $n=| \Xset | = 20$ and the exponent $s=1$.
Similarly, the predictive scheme is identical to the naive scheme when $q$ is small, and it leads to  lower PAoI when the system load is larger. 

\section{Adaptive Empty Buffer Encoding: Extensions}

We have restricted our attention to encoding the empty buffer state using a prefix-free codeword. 
When the codeword length for the null symbol $\ns$ is larger than 1, any new symbol arriving during the transmission of codeword $\enc(\ns)$ will be backlogged in the buffer. 
This is in fact inefficient since the receiver doesn't have to reconstruct the null symbol $\ns$.
Given that the symbol arriving time is not required at the receiver, it is desired to have an encoding scheme that can preempt the transmission of the null symbol and switch the transmitting codeword when a new symbol arrives.
We show this can be achieved if the codeword for the null symbol $\ns$ shares a common prefix with the codeword for the new symbol.

Fig.~\ref{fig:adaptive} demonstrates an example of preempting the transmission of null symbol $\ns$ adaptively using the prefix-free codebook in Table \ref{table:protocol}.
Starting from $t=0$, no symbol arrives to the encoder and the buffer remains empty.
Thus, the buffer starts sending the codeword $101$ one by one starting from $t=0$.
At $t=2$, the first two bits $10$ is delivered to the decoder and a new symbol $B$ arrives at the same time.
Since the codeword for $B$ is $100$, which shares the common first two bits $10$ with the null symbol $\ns$, the encoder can switch to the transmission of symbol $B$ and send the last bit $0$ in the codeword $\enc(B)=100$.
In this case, symbol $B$ is decoded at $t=3$ and the instantaneous age is then reduced to $1$. 

\begin{figure}[t]
\centering
\centering
\includegraphics[width=0.4\textwidth]{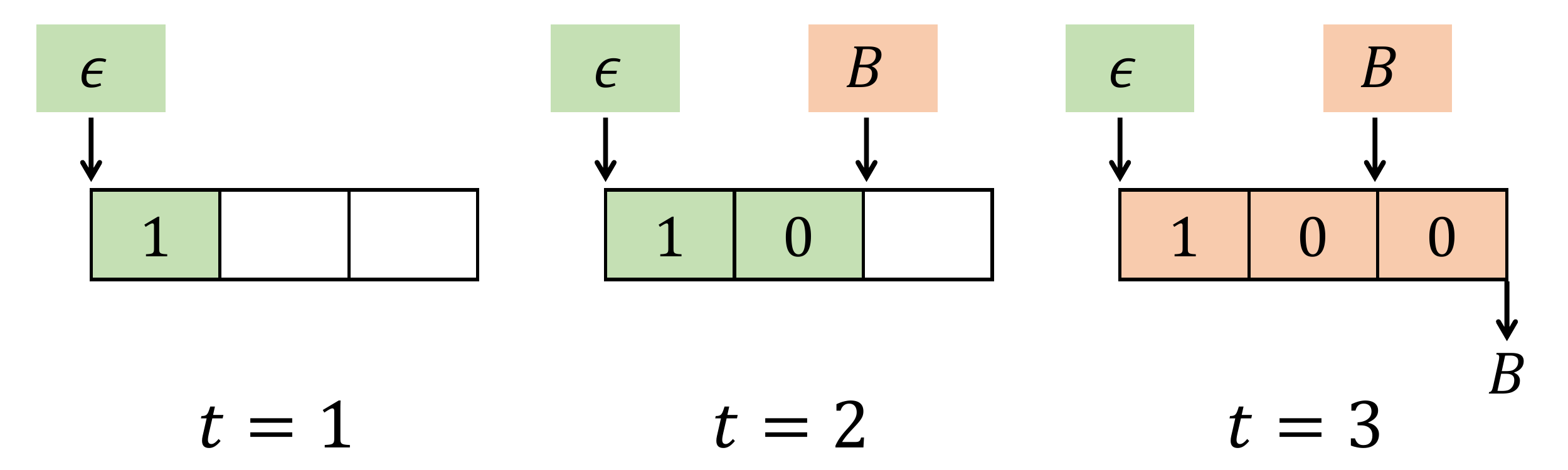}
\caption{An example of adaptive encoding of the empty buffer state.}
\label{fig:adaptive} 
\negspacesmall
\end{figure}

We note that the switch between codewords occurs randomly depending on the probability that $\ns$ and the new symbol are placed in the same branch in the binary code tree. 
Since the null symbol $\ns$ behaves as an estimate of the next arriving symbol, it is expected to assign $\ns$ a long codeword such that it shares a common prefix with most symbols. 
The design of such a coding scheme allowing symbol switching that minimizes the age metric remains as an open problem of interest.

\bibliographystyle{IEEEtran}
\bibliography{refpool}

\end{document}